\documentclass{llncs}
\usepackage[utf8]{inputenc}
\usepackage[margin=1in,footskip=0.25in]{geometry}
\usepackage{amssymb}
\usepackage{amsmath}
\usepackage{enumitem}
\usepackage[linesnumbered,ruled,vlined]{algorithm2e}
\usepackage[dvips]{color}
\usepackage{graphicx}
\graphicspath{{figs/}}
\usepackage{caption}
\DeclareMathOperator*{\argmin}{arg\,min}
\DeclareMathOperator*{\argmax}{arg\,max}
\newcommand{\mt}[1]{\mathrm{#1}}
\newcommand{\mctp}{MINTB}
\newcommand{\parti}{PARTITION}
\newcommand{\sepr}{SP}
\newcommand{\vc}{VC}

\newcommand{\algo}{Algo}
\newcommand{\reals}{\mathbb{R}}
\newcommand{\paths}{\mathcal{P}}
\newcommand{\G}{\mathcal{G}}

\title{New Complexity Results and Algorithms for the Minimum Tollbooth Problem}

\author{Soumya Basu \and Thanasis Lianeas \and
Evdokia Nikolova}

\institute{University of Texas at Austin, USA\\
}

\begin{document}

\maketitle

\begin{abstract}
The inefficiency of the Wardrop equilibrium of nonatomic routing games can be eliminated by placing tolls on the edges of a network so that the socially optimal flow is induced as an equilibrium flow.  A solution where the minimum number of edges are tolled may be preferable over others due to its ease of implementation in real networks. In this paper we consider the minimum tollbooth ($\mctp$) problem, which seeks social optimum inducing tolls with minimum support.  We prove for single commodity networks with linear latencies that the problem is NP-hard to approximate within a factor of $1.1377$ through a reduction from the minimum vertex cover problem. Insights from network design motivate us to formulate a new variation of the problem where, in addition to placing tolls,  it is allowed to remove unused edges by the social optimum.  We prove that this new problem remains NP-hard even for single commodity networks with linear latencies, using a reduction from the partition problem.  On the positive side, we give the first exact polynomial solution to the $\mctp$  problem in an important class of graphs---series-parallel graphs. {Our algorithm solves $\mctp$ by first tabulating the candidate solutions for subgraphs of the series-parallel network and then combining them optimally.}  
       
\end{abstract}

\section{Introduction}


Traffic congestion levies a heavy burden on millions of commuters across the globe.  The congestion cost to the U.S. economy was measured to be \$126 billion in the year 2013 with an estimated increase to \$186 billion by year 2030 \cite{inrix2014cost}.  Currently the most widely used method of mitigating congestion is through congestion pricing, and one of the most common pricing schemes is through placing tolls on congested roads that users have to pay, which makes these roads less appealing and diverts demand, thereby reducing congestion. 

Mathematically, an elegant theory of traffic congestion was developed starting with the work of Wardrop \cite{wardrop1952road} and Beckman et al. \cite{beckmann1956studies}. This theory considered a network with travel time functions that are increasing in the network flow, or the number of users, on the corresponding edges.  Wardrop differentiated between two main goals: (1) user travel time is minimized, and (2) the total travel time of all users is minimized.  This led to the investigation of two different resulting traffic assignments, or flows, called a Wardrop equilibrium and a social or system optimum, respectively.  It was understood that these two flows are unfortunately often not the same, leading to tension between the two different objectives.  Remarkably, the social optimum could be interpreted as an equilibrium with respect to modified travel time functions, that could in turn be interpreted as the original travel time functions plus tolls. 

Consequently, the theory of congestion games developed a mechanism design approach to help users routing along minimum cost paths reach a social optimum through a set of optimal tolls that would be added to (all) network edges.  Later, through the works of Bergendorff et al.~\cite{bergendorff1997congestion} and Hearn \& Ramana~\cite{hearn1998solving}, it was understood that the set of optimal tolls is not unique and there has been work in diverse branches of literature such as algorithmic game theory, operations research and transportation on trying to limit the toll cost paid by users by limiting the number of tolls placed on edges.


\paragraph{Related Work}
The natural question of what is the minimum number of edges that one needs to place tolls on so as to lead selfish users to a social optimum, was first raised by Hearn and Ramana~\cite{hearn1998solving}. The problem was introduced as the minimum tollbooth ($\mctp$) problem and was formulated as a mixed integer linear program. 
This initiated a series of works which led to new heuristics for the problem.  One heuristic approach is based on genetic algorithms \cite{harwood2005genetic,bai2008evolutionary,buriol2010biased}. In 2009, a combinatorial benders cut based heuristics was proposed by Bai and Rubin~\cite{bai2009combinatorial}. The following year, Bai et al. proposed another heuristic algorithm based on LP relaxation using a dynamic slope scaling method \cite{bai2010heuristic}. More recently, Stefanello et al. \cite{stefanello2013minimization} have approached the problem with a modified genetic algorithm technique. 
 
The first step in understanding the computational complexity of the problem was by Bai et al.~\cite{bai2010heuristic} who proved that $\mctp$ in multi commodity networks is NP-hard via a reduction from the minimum cardinality multiway cut problem~\cite{garey2002computers}. 
In a related direction,  
Harks et al.~\cite{harks2008computing} addressed the problem of inducing a predetermined flow, not necessarily the social optimum, as the Wardrop equilibrium, and showed that this problem is APX-hard, via a reduction from length bounded edge cuts \cite{baier2010length}.   
Clearly, $\mctp$ is a special case of that problem and it can be deduced that the hardness results of Harks et al.~\cite{harks2008computing} do not carry forward to the $\mctp$ problem. It should be noted that a related, but distinct, problem has been studied in the algorithmic game theory literature, where one constraints the set of edges on which to impose tolls and seeks tolls that induce the best possible flow under that constraint \cite{hoefer2008taxing,bonifaci2011efficiency}.
 
The latest work stalls at this point leaving open both the question of whether approximations for multi commodity networks are possible,  and what the hardness of the problem is for single commodity networks or for any meaningful subclass of such networks. 
 

\paragraph{Our contribution}
In this work, we make progress on this difficult problem by deepening our understanding on what can and cannot be computed in polynomial time. In particular, we make progress in both the negative and positive directions by providing NP-hardness and hardness of approximation results for the single commodity network, and a polynomial-time exact algorithm for computing the minimum cardinality tolls on series-parallel graphs.

Specifically, we show in Theorem \ref{thm:hardmctp} that  minimum tollbooth problem for single commodity networks and linear latencies is hard to approximate to within a factor of $1.1377$, presenting the first hardness of approximation result for the $\mctp$ problem. 

Further, motivated by the observation that removing or blocking an edge in the network bears much less cost compared to the overhead of toll placement, we ask: if all unused edges under the social optimum are removed, can we solve $\mctp$ efficiently? The NP-hardness result presented in Theorem \ref{thm:hardmctp-ue} for $\mctp$ in single commodity networks with only used edges, settles it negatively, yet the absence of a hardness of approximation result creates the possibility of a polynomial time approximation scheme upon future investigation.

Observing that the Braess structure is an integral part of both NP-hardness proofs we seek whether positive progress is possible in the problem in series-parallel graphs. We propose an exact algorithm for series-parallel graphs with $\mathcal{O}(m^3)$ runtime,  $m$ being the number of edges. Our algorithm provably (see Theorem \ref{thm:correct}) solves the $\mctp$ problem in series-parallel graphs, giving the first exact algorithm for $\mctp$ on an important sub-class of graphs.

\section{  Preliminaries and Problem Definition} \label{sec:prob}
We are given a directed graph $G(V, E)$ with edge delay or latency functions $(\ell_e)_{e \in E}$ and demand $r$ that needs to be routed between a source $s$ and a sink $t$.  We shall abbreviate an instance of the problem by the 
tuple $\G = (G(V, E), (\ell_e)_{e \in E}, r)$. For simplicity, we usually omit the latency functions, and refer to the instance as $(G, r)$. 
The function $\ell_e : \reals_{\geq 0} \rightarrow \reals_{\geq 0}$ is a non-decreasing cost function
associated with each edge $e$. 
Denote the (non-empty) set of simple $s-t$ paths in $G$ by $\paths$. 

\smallskip\noindent{\bf Flows.}
Given an instance $(G, r)$, a (feasible) \emph{flow} $f$ is a non-negative vector indexed by the set of feasible $s-t$ paths $\paths$ such that $\sum_{p \in \paths} f_p = r$. For a flow $f$, let $f_e = \sum_{p: e \in p} f_p$ be the amount of flow that $f$ routes on each edge $e$. 
An edge $e$ is used by flow $f$ if $f_e > 0$, and a path $p$ is used by flow $f$ if it has strictly positive flow on all of its edges, namely $\min_{e \in p} \{ f_e \} > 0$. 
Given a flow $f$, the cost of each edge $e$ is $\ell_e(f_e)$ and the cost of path $p$ is $\ell_p(f) = \sum_{e \in p} \ell_e(f_e)$. 
%

\smallskip\noindent{\bf Nash Flow.}
A flow $f$ is a \emph{Nash (equilibrium) flow}, if it routes all traffic on minimum latency paths. Formally, $f$ is a Nash flow if for every path $p \in \paths$ with $f_p > 0$, and every path $p' \in \paths$, $\ell_p(f) \leq \ell_{p'}(f)$. 
Every instance $(G, r)$ admits at least one Nash flow, and the players' latency is the same for all Nash flows %
(see e.g., \cite{RouB05}).


\smallskip\noindent{\bf Social Cost and Optimal Flow.}
The \emph{Social Cost} of a flow $f$, denoted $C(f)$, is the total latency
$
 C(f) = \sum_{p \in \paths} f_p \ell_p(f) = \sum_{e \in E} f_e \ell_e(f_e)
$\,.
The \emph{optimal} flow of an instance $(G, r)$, denoted $o$, minimizes the total latency among all feasible flows. 
%

\smallskip
 In general, the Nash flow may not minimize the social cost. As discussed in the introduction, one can improve the social cost at equilibrium by assigning tolls to the edges. \\
 \noindent{\bf Tolls and Tolled Instances.} A set of \emph{tolls} is a vector $\Theta=\{\theta_e\}_{e\in E}$ such that the toll for each edge is nonnegative: $\theta_e\geq0$. We call \emph{size} of $\Theta$ the size of the support of $\Theta$, i.e., the number of edges with strictly positive tolls, $|\{e:\theta_e>0\}|$.  Given an instance $\G = (G(V, E), (\ell_e)_{e \in E}, r)$ and a set of tolls $\Theta$, we denote the tolled instance by $\G^{\theta}= (G(V, E), (\ell_e+\theta_e)_{e \in E}, r)$. For succinctness, we may also denote the tolled instance by $(G^{\theta},r)$. We call a set of tolls,  $\Theta$, \emph{opt-inducing} for an instance  $\G$ if the optimal flow in $\mathcal{G}$  and the Nash flow in $\G^{\theta}$  coincide. 

\smallskip
Opt-inducing tolls need not be unique. Consequently, a natural problem is to find a set of optimal tolls of minimum size, which is the problem we consider here.

\begin{definition}[Minimum Tollbooth problem ($\mctp$)]\\
 Given instance $\G$ {and an optimal flow $o$}, find an opt-inducing toll vector $\Theta$ such that the support of $\Theta$ is less than or equal to the support of any other opt-inducing toll vector.
 \end{definition}

The following definitions are needed  for section \ref{sec:algo}. \\
%
%
\noindent{\bf Series-Parallel Graphs.}
A directed $s-t$ multi-graph is \emph{series-parallel} if it consists
of a single edge $(s, t)$ or from two series-parallel
graphs with terminals $(s_1, t_1)$ and $(s_2, t_2)$ composed either in
series or in parallel. In a \emph{series composition}, $t_1$ is identified
with $s_2$, $s_1$ becomes $s$, and $t_2$ becomes $t$.  In a \emph{parallel
composition}, $s_1$ is identified with $s_2$ and becomes $s$, and $t_1$ is
identified with $t_2$ and becomes $t$.


\smallskip
A Series-Parallel ($\sepr$) graph $G$ with $n$ nodes and $m$ edges can be efficiently represented using a parse tree decomposition of size $\mathcal{O}(m)$, which can be constructed in time $\mathcal{O}(m)$  due to Valdes et al.~\cite{valdes1979recognition}.

\noindent{\bf Series-Parallel Parse Tree.}
A series-parallel parse tree $T$ is a rooted binary tree representation of a given $\sepr$ graph $G$ that is defined using the following properties:
\begin{enumerate}
\item Each node in the tree $T$ represents a $\sepr$ subgraph $H$ of $G$, with the root node representing the graph $G$. 
\item There are three type of nodes: `series' nodes, `parallel' nodes, which have two children each, and the `leaf' nodes which are childless. 
\item A `series' (`parallel')  node represents the SP graph $H$ formed by the `series combination' (`parallel combination') of its two children $H_1$ and $H_2$.
\item The `leaf' node represents a parallel arc network, namely one with two terminals $s$ and $t$ and multiple edges from $s$ to $t$. 
\end{enumerate}   
 
For convenience, when presenting the algorithm, we allow `leaf' nodes to be multi-edge/parallel-arc networks. This will not change the upper bounds on the time complexity or the size of the parse tree.

\section{Hardness Results for $\mctp$} \label{sec:hardness}

In this section we provide hardness results for $\mctp$. We study two versions of the problem. The first one considers arbitrary instances while the second considers arbitrary instances where the optimal solution uses all edges, i.e. $\forall e\in E:o_e>0$. Recall that the motivation for separately investigating the second version comes as a result of the ability of the network manager to make some links unavailable.

\subsection{Single-commodity network with linear latencies}
We give hardness results on finding and approximating the solution of $\mctp$ in general instances with linear latencies. In Theorem \ref{thm:hardmctp} we give an inapproximability result by a reduction from a Vertex Cover related NP-hard problem and as a corollary (Corollary \ref{cor:mctpnphard}) we get the NP-hardness of $\mctp$ on single commodity networks with linear latencies. The construction of the network for the reduction is inspired by the NP-hardness proof of the length bounded cuts problem in \cite{baier2010length}.

\begin{theorem}
For  instances with linear latencies, it is NP-hard to approximate the  solution of $\mctp$ by a factor of less than $1.1377$. 
\label{thm:hardmctp}
\end{theorem}
\begin{proof}

The proof is by a reduction from an NP-hard variant of Vertex Cover ($\vc$) due to  
Dinur and Safra \cite{dinur2005hardness}. 
Reminder: a Vertex Cover of  an undirected graph $G(V,E)$ is a set $S\subseteq V$ such that  $\forall \{u,v\}\in E:S\cap \{u,v\} \neq \emptyset$.

Given an instance $\mathcal{V}$ of $\vc$ we are going to construct an instance $\G$ of $\mctp$  which will give a one-to-one correspondence (Lemma \ref{lemm:hardmctp:opteq}) between  Vertex Covers in $\mathcal{V}$  and opt-inducing tolls in  $\G$. The inapproximability result will follow from that correspondence and an inapproximability result concerning Vertex Cover by \cite{dinur2005hardness}. We note that we will not directly construct the instance of $\mctp$. First, we will construct a graph with edge costs that are assumed to be the costs of the edges (used or unused) under the optimal solution  and then we are going to assign linear cost functions and demand that makes the edges under the optimal solution to have costs equal to the predefined costs.


We proceed  with the construction. Given an instance $G_{vc}(V_{vc},E_{vc})$ of $\vc$, with $n_{vc}$ vertices and $m_{vc}$ edges, we construct a directed single commodity network 
$G(V,E)$ with source $s$ and sink $t$ as follows:
\begin{enumerate}
\item For every vertex $v_i\in V_{vc}$ create gadget graph $G_i(V_i,E_i)$, with $V_i=\{a_i, b_i, c_i, d_i\}$ and $E_i=\{(a_i,b_i), (b_i,c_i),$ $ (c_i,d_i), (a_i,d_i)\}$, and assign costs equal to $1$ for edges $e_{1,i}=(a_i,b_i)$ and $e_{3,i}=(c_i,d_i)$, $0$ for edge $e_{2,i}=(b_i,c_i)$, and 3 for edge $e_{4,i}=(a_i,d_i)$. All edges $e_{1,i}$ ,$e_{2,i}$,$e_{3,i}$ and $e_{4,i}$ are assumed to be used.

\item For each edge $e_k=\{v_i,v_j\}\in E_{vc}$ add edges $g_{1,k}=(b_i,c_j)$ and $g_{2,k}=(b_j,c_i)$ with cost $0.5$ each. Edges $g_{1,k}$ and $g_{2,k}$ are assumed to be unused.

\item Add source vertex $s$ and sink vertex $t$ and for all $v_i\in V_{vc}$ add edges $s_{1,i}=(s,a_i)$ and $t_{1,i}=(d_i,t)$ with $0$ cost, and edges $s_{2,i}=(s,b_i)$ and $t_{2,i}=(c_i,t)$ with cost equal to $1.5$. Edges $s_{1,i}$ and $t_{1,i}$ are assumed to be used and edges $s_{2,i}$ and $t_{2,i}$ are assumed to be unused.
\end{enumerate}

     
The construction is shown in figure \ref{fig:hardmctp} where the solid lines represent used edges and dotted lines represent unused edges. The whole network consists of $(2+4n_{vc})$ nodes and $(8n_{vc}+2m_{vc})$ edges,  therefore, it can be constructed  in polynomial time, given $G_{vc}$.

\begin{figure}
\centering
\includegraphics[width=0.7\linewidth]{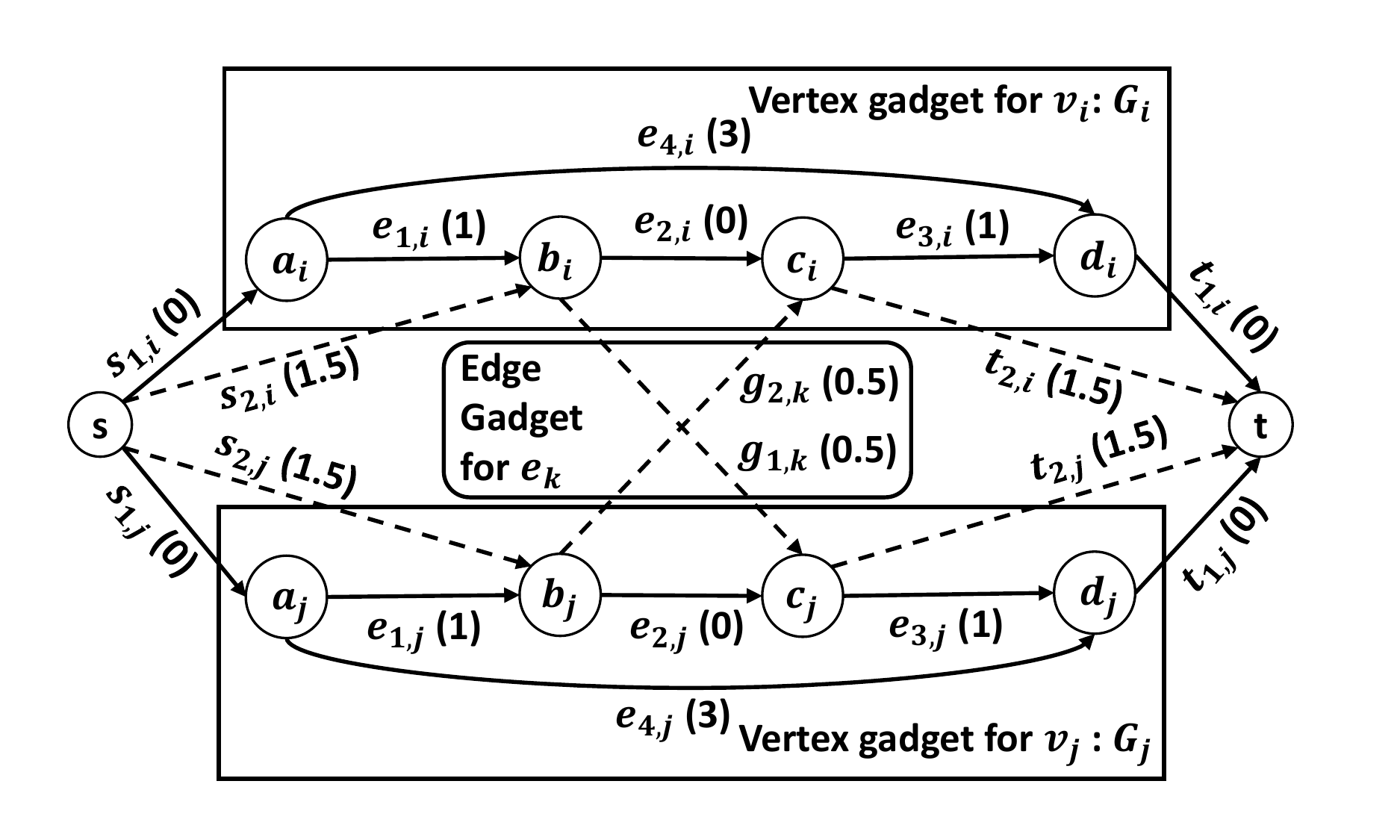}
\caption{Gadgets for the reduction from $\vc$ to $\mctp$.  The pair of symbols on  each  edge corresponds to the name and the cost of the edge respectively. Solid lines represent used edges and dotted lines represent unused edges.}
\label{fig:hardmctp}
\end{figure}

We go on to prove the one-to-one correspondence lemma. 

\begin{lemma} \label{lemm:hardmctp:opteq}
(I) If there is a Vertex Cover in $G_{vc}$ with cardinality $x$, then there are opt-inducing  tolls for  $G$ of size $n_{vc}+x$. \\
(II) If there are opt-inducing  tolls for $G$ of size $n_{vc}+x$, then there is a Vertex Cover in $G_{vc}$ with cardinality $x$.
\end{lemma}

\begin{proof}

We start with some observations:\\
(a) Any solution of $\mctp$ for any $i$ should put a toll on at least one of $e_{1,i}$, $e_{2,i}$ or $e_{3,i}$, as paths $a_i-d_i$ and $a_i-b_i-c_i-d_i$ are both assumed to be used and thus should have equal costs at equilibrium. Thus at least $n_{vc}$ edges are needed to be tolled.\\
(b) For any solution of $\mctp$, if in $G_i$ edge $e_{2,i}$ gets no toll then there must be at least two other edges adjacent to some vertex of $G_i$ that get a toll. That is because if $e_{1,i}$ gets a  toll $< 1$ then $e_{3,i}$ should get a toll in order for the paths inside $G_i$ to be equal, and if $e_{1,i}$ gets a toll $>0.5$ edge $s_{2,i}$ should get a toll in order for path $s-a_i-b_i$ to cost at most as the unused path $s-b_i$.\\
(c) Under any solution of $\mctp$,  for any edge $(v_i,v_j)\in E_{vc}$, at least $3$ edges that are adjacent to some vertex of $G_i$ or $G_j$ should get a toll. To see this first note that the unused $s-t$ paths $p_{ij}=s-a_i-b_i-c_j-d_j-t$ and $p_{ji}=s-a_j-b_j-c_i-d_i-t$ have cost equal to $2.5$ while there are used $s-t$ paths, e.g. $s-a_i-d_i-t$, with cost equal to $3$. Thus at least one  edge in  each of $p_{ij}$ and $p_{ji}$ should get a toll. Now observe that by observation (b), putting tolls on exactly one edge in each of $G_i$ and $G_j$ (one in each is necessary by observation (a)) implies that these edges will be 
$e_{2,i}$ and $e_{2,j}$ (by observation (b)) and these edges  belong to neither in $p_{ij}$ nor in $p_{ji}$, which further implies that more than $2$ edges will be needed to get a toll.\\

(I): Assume that there is a Vertex Cover $A$ of $G_{vc}$ that has cardinality $|A|=x$. We are going to give an opt-inducing toll placement for $G$ of size $n_{vc}+x$. 

For each $i:v_i\in A$ assign toll equal to $0.5$ to $e_{1,i}$ and $e_{3,i}$ and for each $i:v_i\in V\setminus A$ assign toll equal to $1$ to edge $e_{2,i}$. In this way: (i) All paths inside $G_i$ have cost equal to $3$ and thus all used $s-t$ paths have cost equal to  $3$. (ii) Paths $s-a_i-b_i$ and $c_i-d_i-t$ have cost less than that of paths $s-b_i$ and $c_i-t$ respectively, as needed. Finally (iii) for each edge $(v_i,v_j)\in E_{vc}$, used path $s-a_i-b_i-c_i$ has cost less or equal to that of unused path $s-a_j-b_j-c_i$, used path $s-a_j-b_j-c_j$ has cost less or equal to that of unused path $s-a_i-b_i-c_j$, used path $b_j-c_j-d_j-t$ has cost less or equal to that of unused path $b_i-c_j-d_j-t$, and used path $b_i-c_i-d_i-t$ has cost less or equal to that of unused path $b_j-c_i-d_i-t$. The latter holds because  either  $v_i\in A$ or $v_j\in A$, and this implies that a toll equal to $0.5$ has been assigned to $e_{1,i}$ and $e_{3,i}$ or to $e_{1,j}$ and $e_{3,j}$ (or both). 

The above  imply that we indeed have an equilibrium. The edges that got toll by the described assignment have cardinality $2|A|+|V\setminus A|=n_{vc}+x$. 

(II)  Assume that for $G$ we have  opt-inducing toll vector $\Theta=\{\tau_e\}_{e\in E}$ of size $n_{vc}+x$ ($x>0$, recall observation (a)), and let $B$ be the set of edges of $G$ that get positive toll, i.e. $B=\{e\in E: \tau_e>0\}$. We are going to find a Vertex Cover of cardinality $x$. For this we will need the following definition. We say that edge $e$ \emph{touches} $G_i$ if one of its endpoints is adjacent to a vertex in $G_i$.

By this definition, let $A$ be the set of all vertices $v_i$ of $V_{vc}$ that have their corresponding $G_i$ being touched by at least two edges belonging in $B$, i.e. $A=\{ v_i \in V_{vc}:\exists e_1,e_2\in B \,\, that \,\, touch \,\, G_i \}$. We will prove that $A$ is a Vertex Cover of $G_{vc}$ with the desired cardinality.

Assume that there is an edge $\{v_i,v_j\}\in E_{vc}$ that is not covered by $A$. Not covered by $A$ implies, by definition of $A$ and observation (a), that each of $G_i$ and $G_j$ are touched by exactly one edge in $B$. This is a contradiction as, by observation (c), $G_i$ and $G_j$ should have at least $3$ edges touching them.

To bound the cardinality of $A$ we first note that by observation (a) all $G_i$'s are touched by at least $1$ edge. Thus, the set of $G_i$'s that are touched by more than $1$ edge have cardinality $\leq |B|-n_{vc}$ which implies $|A|\leq |B|-n_{vc}=x$.  Now $|A|\leq x$  implies that there is a Vertex Cover with $x$ vertices (if $|A|<x$ just add extra vertices to reach equality). \qed
 
\end{proof}

Statement (I) in the above lemma directly implies that if the minimum Vertex Cover of $G_{vc}$ has cardinality $x$ then the optimal solution of the $\mctp$ instance has size at most $n_{vc}+x$.

From the proof of Theorem $1.1$ in \cite{dinur2005hardness} we know that there exist instances $G_{vc}$ where it is NP-hard to distinguish between the case where we can find a Vertex Cover of size 
$n_{vc}\cdot(1 - p + \epsilon)$, and the case where any vertex cover has size at least 
$n_{vc}\cdot(1 - 4p^3+3p^4 - \epsilon)$, for any positive $\epsilon$ and $p = (3 - \sqrt{5})/2$. We additionally know that the existence of a Vertex Cover with cardinality in between the gap implies the  existence of a Vertex Cover of cardinality $n_{vc}\cdot(1 - p + \epsilon)$.\footnote{the instance they create will  have either  a Vertex Cover of cardinality $n_{vc}\cdot(1 - p + \epsilon)$ or all Vertex Covers with cardinality $\geq n_{vc}\cdot(1 - 4p^3+3p^4 - \epsilon)$ } 
%
%

{Assuming that we reduce from such an instance of $\vc$}, the {above} result implies that it is NP-hard to approximate  $\mctp$  within a factor of 
$1.1377<\frac{2 - 4p^3+3p^4 - \epsilon}{2 - p + \epsilon}$ (we chose an $\epsilon$ for inequality to hold).
 To reach a contradiction assume the contrary, i.e. there exists a $\beta$-approximation algorithm $\algo$ for $\mctp$, where $\beta \mathbf{\leq} 1.1377<\frac{2 - 4p^3+3p^4 - \epsilon}{2 - p + \epsilon}$. 
By Lemma \ref{lemm:hardmctp:opteq} statement (I), if there exist a Vertex Cover of cardinality $\hat{x}=n_{vc}\cdot(1 - p + \epsilon)$ in $G_{vc}$, then the cardinality in an optimal solution to $\mctp$ on the corresponding instance is 
 $OPT \leq n_{vc}+\hat{x}$. Further, $Algo$ produces an opt-inducing tolls with size $n_{vc}+y$, from which we can get a Vertex Cover of cardinality $y$ in the same way as we did inside the proof of statement (II) of Lemma \ref{lemm:hardmctp:opteq}. Then by the approximation bounds and using $\hat{x}=n_{vc}\cdot(1 - p + \epsilon)$ we get
 $$\frac{n_{vc}+y}{n_{vc}+\hat{x}} \leq \frac{n_{vc}+y}{OPT}\leq\beta < \frac{2 - 4p^3+3p^4 - \epsilon}{2 - p + \epsilon}\Rightarrow 1+\frac{y}{n_{vc}}<2 - 4p^3+3p^4 - \epsilon\Rightarrow y< (1 - 4p^3+3p^4 - \epsilon)n_{vc}$$
The last inequality would answer the question whether there exist a Vertex Cover with size $n_{vc}\cdot(1 - p + \epsilon)$,  as we started from an instance for which we additionally know that the existence of a Vertex Cover with cardinality $y< (1 - 4p^3+3p^4 - \epsilon)n_{vc}$ implies the  existence of a Vertex Cover of cardinality $n_{vc}\cdot(1 - p + \epsilon)$.

What is left for concluding the proof is to define the linear cost functions and the demand so that at optimal solution all edges have costs equal to the ones defined above.

Define the demand to be $r=2n_{vc}$ and assign: for every $i$ the cost functions $\ell_0(x)=0$ to edges $s_{1,i}$, $t_{1,i}$ and $e_{2,i}$, the cost function $\ell_1(x)=\frac{1}{2}x+\frac{1}{2}$ to edges $e_{1,i}$ and $e_{3,i}$, the cost function $\ell_2(x)=1.5$ to edges $s_{2,i}$ and $t_{2,i}$, and the cost function  $\ell_3(x)=3$ to edge $e_{4,i}$, and for each $k$, the cost function  $\ell_4(x)=0.5$  to edges $g_{1,k}$ and $g_{2,k}$.
The optimal solution will assign for each $G_i$ one unit of flow to path $s-a_i-b_i-c_i-d_i-t$ and one unit of flow to $s-a_i-d_i-t$. This makes the costs of the edges to be as needed, as the only non constant cost is $\ell_1$ and $\ell_1(1)=1$.

To verify that this is indeed an optimal flow, one can assign to each edge $e$ instead of its cost function, say $\ell_e(x)$,  the cost function $\ell_e(x)+x\ell_e'(x)$. The optimal solution in the initial instance should be an equilibrium for the instance with the pre-described change in the cost functions (see e.g. \cite{roughgarden2003price}). This will hold here as under the optimal flow and with respect to the new cost functions the only edges changing cost will be $e_{1,i}$ and $e_{3,i}$, for each $i$, and that new cost will be $1.5$ ($\ell_1(1)+1\ell'_1(1)=1.5$).\qed

\end{proof}

Consequently, we obtain the following corollary.

\begin{corollary}\label{cor:mctpnphard}
For  single commodity networks with linear latencies, $\mctp$ is NP-hard.
\end{corollary}

\begin{proof}
By following the same reduction, by Lemma \ref{lemm:hardmctp:opteq} we get that solving $\mctp$  in $G$ gives  the solution to $\vc$ in $G_{vc}$ and vice versa. Thus, $\mctp$ is NP-hard.\qed
\end{proof}

\subsection{Single-commodity network with linear latencies and all edges under use}
In this section we turn to study $\mctp$ for instances where all edges are used by the optimal solution. 
Note that this case is not captured by Theorem \ref{thm:hardmctp}, as in the reduction given for proving the theorem, 
the existence of unused paths in  network $G$  was crucially exploited. 
Nevertheless, $\mctp$ remains $NP$ hard  for this case. 

\begin{theorem}
For  instances with linear latencies, it is NP-hard to solve $\mctp$ even if all edges are used by the optimal solution.
\label{thm:hardmctp-ue}
\end{theorem}


\begin{proof}


The proof comes by a reduction from the partition problem ($\parti$) which is well known to be NP-complete (see e.g. \cite{garey2002computers}). $\parti$   is: Given a multiset  $S=\{\alpha_1,\alpha_2,\dots,\alpha_n\}$ of positive integers, decide (YES or NO) whether there exist a partition of $S$ into sets $S_1$ and $S_2$ such that $S_1\cap S_2=\emptyset$ and $\sum_{\alpha_i\in S_1} \alpha_i =\sum_{\alpha_j \in S_2} \alpha_j=\frac{\sum_{i=1}^{n}\alpha_i}{2}$.

 Given an instance of $\parti$ we will construct an instance of $\mctp$ with used edges only and show that getting the  optimal solution for $\mctp$ solves $\parti$. Though, we will not directly construct the instance. First we will construct a graph with edge costs that are assumed to be the costs of the edges under the optimal solution and then we are going to assign linear cost functions and demand that makes the edges under the optimal solution to have costs equal to the predefined costs.  
 For these costs we will prove that  if the answer to $\parti$ is YES, then the  solution to $\mctp$ puts tolls to $2n$ edges and if the answer to $\parti$ is NO then the  solution to $\mctp$ puts tolls to more than $2n$ edges. Note that the tolls that will be put on the edges should make all s-t paths of the $\mctp$ instance having equal costs, as all of them are assumed to be used. 
 
 Next, we construct the graph of the reduction together with the costs of the edges.
Given the multi-set $S=\{\alpha_1,\alpha_2,\dots,\alpha_n\}$  of $\parti$, with $\sum_{i=1}^n \alpha_i=2B$,  construct the $\mctp$ instance graph $G(V,E)$, with source $s$ and sink $t$, in the following way:

\begin{enumerate}
\item For each $i$, construct graph $G_i=(V_i,E_i)$, with $V_i=\{u_i, w_i,x_i,v_i\}$ and 
$E_i=\{(u_i,w_i), (w_i,v_i), (u_i,x_i), (x_i,v_i),$  $ (w_i,x_i), (w_i,x_i)\}$. Edges $a_i=(u_i,w_i)$ and  $b_i=(x_i,v_i)$ have cost equal to $\alpha_i$, edges $c_{1,i}=(w_i,x_i)$ and $c_{2,i}=(w_i,x_i)$ have cost equal to $2\alpha_i$ and edges  $q_i=(w_i,v_i)$ and $g_i=(u_i,x_i)$ have cost equal to $4 \alpha_i$.
\item For $i=1$ to $n-1$ identify $v_i$ with $u_{i+1}$. Let the source vertex be $s=u_1$ and the sink vertex  be $t=v_n$.
\item Add  edge $h=(s,t)$ to connect $s$ and $t$ directly with cost equal to $11B$.
\end{enumerate}


The constructed graph is presented in figure \ref{fig:hardmctpue}. It is consisted of $(3n+1)$ vertices and $6n+1$ edges and thus can be created in polynomial time, given $S$.  

\begin{figure}
\centering
\includegraphics[width=0.9\linewidth]{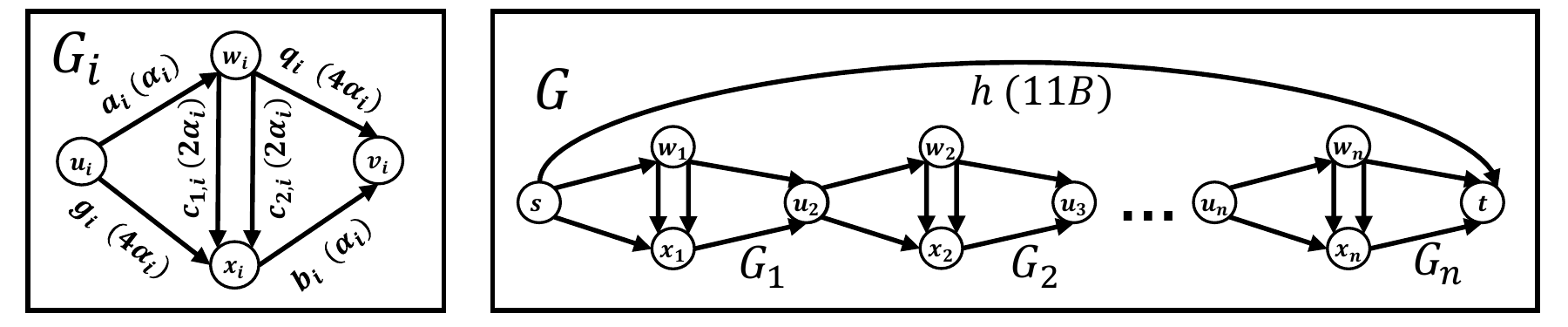}
\caption{The graph for $\mctp$, as it arises from $\parti$. The pair of symbols on  each of $G_i$'s edges correspond to the name and the cost of the edge respectively.}
\label{fig:hardmctpue}
\end{figure}

We establish the one-to-one correspondence between the two problems in the following lemma.

\begin{lemma}
(I) If the answer to $\parti$ on $S$ is YES then the size of opt-inducing tolls for $G$ is equal to $2n$. \\
(II) If the answer to $\parti$ on $S$ is NO then the size of opt-inducing tolls for $G$ is strictly greater than $2n$.
\end{lemma}
\begin{proof}
We now give some useful observations needed for the proof: \\
(a) All $s-t$ paths, through $G_i$'s will cost at most $10B$ (at most $5\alpha_i$ in each $G_i$) and under  any solution, all $s-t$ paths in $G$, should cost at least  $11B$ due to (used) edge $h$.\\
(b)  In a solution of $\mctp$ all  paths should have equal costs. A necessary condition for that is that for every $G_i$ all $u_i-v_i$ paths from $u_i$ to $v_i$ should have equal cost. For this to happen,  tolls on at least $2$ edges for each $G_i$ are needed, as the zig zag $u_i-v_i$ paths (paths through vertical edges) should get cost equal to the parallel $u_i-v_i$ paths (paths avoiding vertical edges).
Therefore, in the optimal solution we need at least $2n$ edges to get tolled. Furthermore, there are only $2$ ways in which we can make the costs in $G_i$ equal  by putting tolls in exactly $2$ edges: we may place toll $\alpha_i$ on both $a_i$ and $b_i$ which makes all the paths from $u_i$ to $v_i$ equal to $6\alpha_i$ or we may place toll $\alpha_i$ on edges $c_{1,i}$ and $c_{2,i}$ which makes all $u_i$ to $v_i$ paths equal to $5\alpha_i$. \\

Going on with the proof, if the answer to $\parti$ is YES then the solution to $\mctp$ is putting tolls to exactly $2n$ edges. To see this first note that by observation (b) above, it is necessary for $2n$ edges  to get a toll. Now, let $S_1$ and $S_2$ be the solution of $\parti$, i.e. $S_1\cap S_2=\emptyset$ and $\sum_{\alpha_i\in S_1} \alpha_i =\sum_{\alpha_j \in S_2} \alpha_j=B$, and let  $I_1,I_2\subset \{1\ldots n\}$ be the corresponding sets of indices such that $\forall i\in I_1: \alpha_i\in S_1$ and  $\forall i\in I_2: \alpha_i\in S_2$. A solution of $\mctp$ that puts a toll in exactly $2n$ edges is the following: for each $i\in I_1$ put a toll $\alpha_i$ on both $a_i$ and $b_i$ and for  each $i\in I_2$ put a toll $\alpha_i$ on both $c_{1,i}$ and $c_{2,i}$. This results to an equilibrium as all  paths in $G_i$ for  $i\in I_1$ will have cost equal to $6\alpha_i$, all paths in $G_i$ for  $i\in I_2$ will have cost equal to $5\alpha_i$ and all $s-t$ paths will have cost equal to $11B$. 
The latter comes by the fact that all $s-t$ paths through the $G_i$'s will have cost $\sum_{i\in I_1}6\alpha_i+\sum_{i\in I_2}5\alpha_i=6B+5B=11B$ and edge $h$ has cost $11B$.

 
If the answer to $\parti$ is NO then the solution to $\mctp$ is putting tolls to more than $2n$ edges. Assuming the contrary we will reach a contradiction. Let the solution of $\mctp$ use $2n$ edges (by observation (b) it cannot use less). This, by observation (b), implies  that in each $G_i$ there are exactly $2$ edges that get a toll, which in turn, again by observation (b), implies that there exist sets $I_1$ and $I_2$ such that all paths in $G_i$ for $i\in I_1$ cost equal to $6\alpha_i$ and all paths in $G_i$ for $i\in I_2$ cost equal to $5\alpha_i$. Additionally, solving the problem by putting tolls on exactly $2n$ edges  implies  that edge $h$ does not get any toll and thus all $s-t$ paths cost equal to $11B$, i.e. the cost of edge $h$. These together  give $11B=\sum_{i\in I_1}6a_i+\sum_{i\in I_2}5a_i=\sum_{i\in I_1}a_i+\sum_{i=1}^n 5a_i= \sum_{i \in I_1} a_i+10B$, which implies  $\sum_{i \in I_1} a_i=B$ and contradicts that we have a NO instance of $\parti$. \qed
 
  
\end{proof}

What is left for concluding the proof is to define the linear cost functions and the demand so that at optimal solution all edges have costs equal to the ones defined above.

Define the demand to be $r=4$ and assign the cost function $\ell_h=11B$ to edge $h$ and for each $i$, the cost function $\ell_i^1(x)=\frac{1}{4}\alpha_i x+\frac{1}{2}\alpha_i$ to edges $a_i$, $b_i$, the cost function $\ell_i^2(x)=\alpha_i x+\frac{3}{2}\alpha_i$ to edges $c_{1,i}$ and $c_{2,i}$,  and the constant cost function $\ell_i^3(x)=4\alpha_i$ to edges $q_i$ and $g_i$. The optimal flow then assigns $1$ unit of flow to edge $h$ which has cost $11B$, and the rest $3$ units to the paths through $G_i$. In each $G_i$, $1$ unit  will pass through $a_i-q_i$,  $1$ unit  will pass through $g_i-b_i$, $1/2$ units will pass through $a_i-c_{1,i}-b_i$, and $1/2$ unit will pass through $a_i-c_{2,i}-b_i$. This result to $a_i$ and $b_i$ costing $\alpha_i$, to $c_{1,i}$ and $c_{2,i}$ costing $2\alpha_i$, and to $q_i$ and $g_i$ costing $4\alpha_i$, as needed. 

We verify that it is indeed an optimal flow using technique similar to the one used in Theorem \ref{thm:hardmctp}. Specifically, under the optimal flow and with respect to the new cost functions, $\ell_e(x)+x\ell_e'(x)$: (i) $h$ will cost $11B$, (ii) for each $i$, $a_i$ and $b_i$ will cost $\frac{3}{2}\alpha_i$, $c_{1,i}$ and $c_{2,i}$ will cost $\frac{5}{2}\alpha_i$, and  $q_i$ and $g_i$ will cost $4\alpha_i$, which makes each path in $G_i$ costing $\frac{11}{2}\alpha_i$, and (iii) the cost of any path through the $G_i$'s has cost $\sum_{i=1}^n \frac{11}{2}\alpha_i =11B  $.\qed
\end{proof}

\section{Algorithm for MINTB on Series-Parallel Graphs} \label{sec:algo}
In this section we propose an exact algorithm for  $\mctp$  in series-parallel graphs. {We do so by reducing it to a solution of an equivalent problem defined below.}
 
Consider an instance $\mathcal{G}=\{G(V,E),(\ell_e)_{e\in E},r\}$ of  $\mctp$,  where $G(V,E)$ is a $\sepr$ graph with terminals $s$ and $t$. Since the flow we want to induce is fixed, i.e. the optimal flow $o$, by abusing notation,  let  \emph{length} $\ell_e$ denote $\ell_e(o_e)$, for each $e\in E$, and \emph{used edge-set}, $E_u=\{e\in E:o_e>0\}$, denote the set of used edges under $o$. 
For  $\G$, we define the corresponding  \emph{l-instance} (length-instance) to be $\mathcal{S}(\G)=\{G(V,E),\{l_e\}_{e\in E}, E_u\}$.  We may write simply $\mathcal{S}$, if $\G$ is clear from the context. 
By the  definition below  and the equilibrium definition, Lemma \ref{lemm:induceequiv}  easily follows.



\begin{definition}
{Given an $l$-instance $\mathcal{S}=\{G(V,E),\{l_e\}_{e\in E}, E_u\}$,  
inducing} a length $L$ in $G$ is defined as the process of finding  $\ell'_e \geq \ell_e$, for all $e\in E$, such that when replacing $\ell_e$ with $\ell_e'$: (i) all used $s-t$ paths have length $L$ and (ii) all unused $s-t$ paths have length greater or equal to $L$, where a path is used when all of its edges are used, i.e. they belong to $E_u$. 
\label{def:induce}
\end{definition}



\begin{lemma}
{Consider an instance $\G$ on a $\sepr$ graph $G(V,E)$ with corresponding $l$-instance $\mathcal{S}$. }  $L$ is induced in $G$ with modified lengths $\ell_e'$ if and only if  $\{\ell'_e-\ell_e\}_{e\in E}$ is an opt-inducing toll vector for $\G$.  
\label{lemm:induceequiv}
\end{lemma}

We call edges with $\ell_e'>\ell_e$ tolled edges as well. Under these characterizations, observe that finding a toll vector $\Theta$ that solves $\mctp$ for  instance $\G$ with graph $G$, is equivalent to inducing length $L$ in $G$ with minimum number of tolled edges, where $L$ is the  common equilibrium cost of the used paths in $G^\theta$. In general, this $L$ is not known in advance and it might be greater than $\ell_{max}$, i.e. the cost of the most costly used path in $G$, see e.g. fig. \ref{fig:counter}. Though, for $\sepr$ graphs we prove (Lemma \ref{lemm:SP}) a monotonicity property that ensures  that inducing length $\ell_{max}$ results in less or equal number of tolled edges than  inducing any $\ell'>\ell_{max}$. Our algorithm relies on the above equivalence and induces $\ell_{max}$ with minimum number of tolled edges. 


\begin{figure}[!ht]
\begin{minipage}{0.7\linewidth}
\input{makelistpa}
\end{minipage} 
\hfill
\begin{minipage}{0.3\linewidth}\centering
\includegraphics[width=0.6\linewidth]{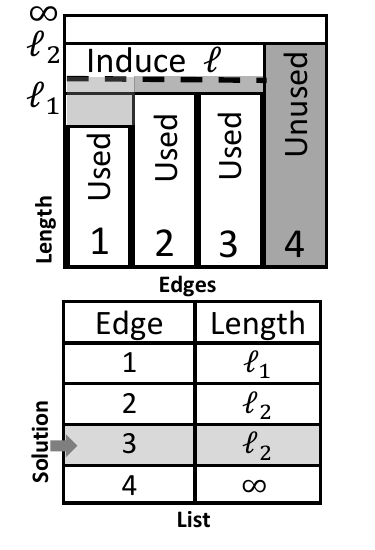}
\captionof{figure}{Example of list}
\label{fig:example}
\end{minipage}
\end{figure}


\subsubsection{Algorithm for parallel link networks:} \label{sssec:algoPA}

Before introducing the algorithm for $\mctp$  on $\sepr$ graphs, we consider the problem of inducing a length $L$ in a parallel {link} network $P$ using minimum number of edges. It is easy to see that all edges with length less than the maximum among used edges, say $\ell_{max}$, should get a toll. Similarly, to induce any length $\ell>\ell_{max}$, all edges with cost less than $\ell$ are required to be tolled. 

Define an `edge-length' pair as the pair $(\eta,\ell)$ such that by using at most $\eta$ edges a length $\ell$ can be induced in a given graph. Based on the above observations we create the `edge-length' pair list, $lst_P$, in Algorithm \ref{alg:makelistpa}.  {By r}eordering the edges in increasing length order, let edge $k$ have length $\ell_k$ for $k=1$ to $m$. Also let there be $i_0$ number of edges with length less than $\ell_{max}$. The list gets the first entry $(i_0,\ell_{max})$ and subsequently for each $i=i_0+1$ to $m$, {gets the entry} $(i,\ell_{i+1})$, where $\ell_{m+1}=\infty$. 

 To induce any length $\ell$, starting from the first `edge-length' pair in list $lst_P$ {we} linearly scan the list until for the first time we encounter the `edge-length' pair with $\eta$ edges and length strictly greater than $\ell$. Clearly $(\eta-1)$ is the minimum number of edges required to induce $\ell$ as illustrated in Figure \ref{fig:example}.

\subsubsection{Algorithm structure:}

The proposed algorithm for $\mctp$, Algorithm \ref{alg:mctp}, proceeds in a recursive manner on a given parse tree $T$ of the $\sepr$ graph $G$ of an $l$-instance $\mathcal{S}$, where we create $\mathcal{S}$ given instance $\mathcal{G}$ and optimal flow $o$. Recall that for each node $v$ of the parse tree we have an associated $\sepr$ subgraph $G_v$ with the terminals $s_v$ and $t_v$. The two children of node $v$, whenever present, represent two subgraphs of $G_v$, namely $G_1$ and $G_2$. Similar to the parallel link graph our algorithm creates an `edge-length' pair list for each node $v$. 

\vspace{-7 pt} 
\paragraph{Central idea}
Beginning with the creation of a list for each leaf node of the parse tree using Algorithm \ref{alg:makelistpa} we keep on moving up from the leaf level to the root level. At every node the list of its two children, $lst_1$ and $lst_2$, are optimally combined to get the current list $lst_v$.
For each `edge-length' pair $(\eta,\ell)$ in a current list we maintain two pointers $(p1,p2)$ to point to the two specific pairs, one each from its descendants, whose combination generates the pair $(\eta,\ell)$. Hence each element in the list of a `series' or `parallel' node $v$ is given by a tuple, $(\eta,\ell,p1,p2)$. 

The key idea in our approach is that the size of the list $lst_v$ for every node $v$, is upper bounded by the number of edges in the subgraph $G_v$. Furthermore, for each series or parallel node, we device  polynomial time algorithms, Algorithm \ref{alg:combineseries} and Algorithm \ref{alg:combineparallel} respectively, which carry out the above combinations optimally.

\vspace{-7 pt}
\paragraph{Optimal list creation}
Specifically, we first compute the number of edges necessary to induce the length of maximum used path between $s_v$ and $t_v$, which corresponds to the first `edge-length' pair in $lst_v$. Moreover, the size of the list is limited by the number of edges necessary for inducing the length 
$\infty$, as computed next. Denoting the first value by $s$ and the latter by $f$, for any `edge-length' pair $(\eta,\ell)$ in $lst_v$, $\eta\in \{s,s+1,\dots,f\}$. 

Considering an $\eta$ in that range we may use $\eta'$ edges in subgraph $G_1$ and $\eta-\eta'$ edges in subgraph $G_2$ to induce some length, which gives a feasible division of $\eta$. Let $\eta'$ induce $\ell_1$ in $G_1$ and $\eta-\eta'$ induce $\ell_2$ in $G_2$. In `series' node the partition induces $\ell=\ell_1+\ell_2$ whereas in `parallel' node it induces $\ell=\min\{\ell_1,\ell_2\}$.

Next we fix the number of edges to be $\eta$ and find the feasible division that maximizes the induced length in $G$ and subsequently a new `edge-length' pair is inserted in $lst_v$. We repeat for all $\eta$, starting from $s$ and ending at $f$. This gives a common outline for both Algorithm \ref{alg:combineseries} and Algorithm \ref{alg:combineparallel}. A detailed description is provided in Theorem \ref{thm:optmakelist}.

\vspace{-7 pt}
\paragraph{Placing tolls on the network}
Once all the lists have been created, Algorithm \ref{alg:placetoll} traverses the parse tree starting from its root node and optimally induces the necessary lengths at every node. At the root node the length of maximum used path in $G$ is induced. 
At any stage, due to the optimality of the current list, given a length $\ell$ that can be induced there exists a unique `edge-length' pair that gives the optimal solution. In the recursive routine after finding this specific pair, we forward the length required to be induced on its two children. For a `parallel' node  the length $\ell$ is forwarded to both of its children, whereas in a `series' node the length is appropriately split between the two. Following the tree traversal the algorithm eventually reaches the leaf nodes, i.e. the parallel link graphs,  where given a length $\ell$ the optimal solution is to make each edge $e$ with length $\ell_e < \ell$ equal to length $\ell$ by placing toll $\ell-\ell_e$. A comprehensive explanation is presented under Lemma \ref{lemm:optplacetoll}.     

\begin{algorithm}[H]
\Indm
\KwIn{$\mctp$ instance $\mathcal{G}=\{G(V,E),(\ell_e)_{e\in E},r\}$ for a $\sepr$ graph $G(V,E)$, An optimal flow: $o$ }
\KwOut{Minimal cardinality tolls $\Theta$}
\Indp
 Create $l$-instance $\mathcal{S}=\{G(V,E),\{l_e\}_{e\in E}, E_u\}$\; 
 Create a Parse tree $T$ for $G(V,E)$\;
 $\ell_{max} \gets $ cost of max $(s,t)$ used path in $G$\;
 Create empty collection of lists $\mathcal{L}$ and update,  $\mathcal{L} \gets \mt{MAKELIST}\left(T\right)$\; \label{alg:mctp:makelist}
 Compute the optimal tolls, $\Theta \gets \mt{PLACETOLL}\left(T, \ell_{max}\right)$\; \label{alg:mctp:placetoll}
 \caption{SolMINTB}
 \label{alg:mctp}
\end{algorithm}

\begin{algorithm}[H]
\Indm
\KwIn{Parse Tree: $T$ with root $r$,  Collection of lists:  $\mathcal{L}$ (Global) }
\KwOut{Processed collection of lists: $\mathcal{L}_{T}=\{lst_v : v \in V \cap T \}$}
\Indp 
 \If{$T$ is a leaf node}{
 $\mt{MAKELISTPL}(T)$\; 
 \Return
 }
 Recur on the children, $\mt{MAKELIST}\left(T{\cdot}p\right)$, $p \in \left\{ 1, 2\right\}$\;
 \eIf{Root node of $T$ is Series}{
 $\mt{COMBINESERIES}\left(lst_{1}, lst_{2}, lst_{r}\right)$\;
 }{
 $\mt{COMBINEPARALLEL}\left(lst_{1}, lst_{2}, lst_{r}\right)$\;
 }
 \caption{MAKELIST}
 \label{alg:makelist}
\end{algorithm}

\begin{algorithm}[H]
\Indm
 \KwIn{Parse Tree: $T$ with root $r$, Length: $\ell_{in}$, Collection of lists: $\mathcal{L}$ (Global) }
 \KwResult{Toll on each edge $\Theta_{T}=\{\theta_e \in \mathbb{R}_{+}: e \in E \cap T \}$}
\Indp 
 \If{$T$ is a leaf node}{
 \For{ each edge $e$ with $\ell_e < \ell_{in}$ in the parallel link }{
  Set toll $\theta_e \gets \ell_{in}-\ell_e$\; \label{alg:placetoll:set_toll}
  }
  \Return
 }
 Select element in list $lst_r$ to induce $\ell_{in}$,
 $opt \gets \argmin_{j} \{lst_{(r,j)}{\cdot}\eta : lst_{(r,j)}{\cdot}\ell \geq \ell_{in}\}$\;
 \eIf{Root node of $T$ is Series}{\label{alg:placetoll:divlen}
 Fix cost on the right sub-tree, $\ell_{2} \gets lst_{(r, opt)}{\cdot}p2{\cdot}\ell$\;
 Fix cost on the left sub-tree, $\ell_{1} \gets \ell_{in}-\ell_{2}$\; 
 }{
 Fix cost on the sub-trees,  $\ell_{1} \gets \ell_{in}$, $\ell_{2} \gets \ell_{in}$\;
 \label{alg:placetoll:divlen2}  
 }
 Recur on root $r_1$ of first children $T{\cdot}1$, $\mt{PLACETOLL}\left(T{\cdot}1, \ell_1\right)$  \;
 Recur on root $r_2$ of second children $T{\cdot}2$, $\mt{PLACETOLL}\left(T{\cdot}2, \ell_2\right)$  \;
 
 \caption{PLACETOLL}
 \label{alg:placetoll}
\end{algorithm}

\newpage

\begin{algorithm}[H]
\Indm
\KwIn{Lists: $lst_{1}$, $lst_{2}$, $lst$.}
\KwOut{Processed List: $lst$.}
\Indp 
  Let $l$ and $r$ be the size of $lst_{1}$ and $lst_{2}$ respectively\; 
  Min number of edges in $lst$, $s\gets lst_{(1,0)}{\cdot}\eta + lst_{(2,0)}{\cdot}\eta$\;  \label{alg:combineseries:s}
  Max number of edges in $lst$,
  $f\gets \min\left\{lst_{(1,l)}{\cdot}\eta + lst_{(2,0)}{\cdot}\eta, 
 lst_{(1,0)}{\cdot}\eta + lst_{(2,r)}{\cdot}\eta\right\}$\; \label{alg:combineseries:f}
 
 \For{$i \gets s$ \textbf{to} $f$}{\label{alg:combineseries:for}
 Find the possible edge divisions for $i$, \label{alg:combineseries:div}
 $\mathcal{I}=\left\{\left(j_{1},j_{2}\right): lst_{(1,j_1)}{\cdot}\eta+lst_{(2,j_2)}{\cdot}\eta=i \right\}$\;
 Find the cost maximizing division\;
    $\left(opt_1,opt_2\right)=\argmax_{\left(j_{1},j_{2}\right) \in \mathcal{I}} 
   \left( lst_{(1,j_1)}{\cdot}\ell+lst_{(2,j_2)}{\cdot}\ell\right)$\; \label{alg:combineseries:maxdiv}
   $ \ell_{i} \gets lst_{(1,opt_1)}{\cdot}\ell+lst_{(2,opt_2)}{\cdot}\ell$\;
 Create the new element $\alpha$,
   	$\left(\alpha{\cdot}\eta, \alpha{\cdot}\ell, \alpha{\cdot}p1, \alpha{\cdot}p2\right) \gets $   
   	$ \left( i, \ell_{i}, lst_{(1,opt_1)}, lst_{(2,opt_2)}\right)$\;
   	Insert $\alpha$ in list $lst$\;
 }
 \caption{COMBINESERIES}
 \label{alg:combineseries}
\end{algorithm}

\begin{algorithm}[H]
\Indm
\KwIn{Lists: $lst_{1}$, $lst_{2}$, $lst$.}
\KwOut{Processed List: $lst$.}
\Indp 
  Let $l$ and $r$ be the size of $lst_{1}$ and $lst_{2}$ respectively\;
  Length of maximum used path in the combined graph \\
  $\ell_{max}\gets \max\{lst_{(1,0)}{\cdot}\ell, lst_{(2,0)}{\cdot}\ell\}$\;
  Min number of edges to induce $\ell_{max}$ in $G_p$ for $p\in\{1,2\}$, \\
  $s_p \gets \min\{lst_{(p,j)}{\cdot}\eta:lst_{(p,j)}{\cdot}\ell\geq \ell_{max}\}$ \; \label{alg:combineparallel:s:sub}  
  Min number of edges in $lst$,
  $s\gets s_{1}+s_{2}$\;\label{alg:combineparallel:s}
  Max number of edges in $lst$,
 $f\gets lst_{(1,l)}{\cdot}\eta+lst_{(2,r)}{\cdot}\eta$\;\label{alg:combineparallel:f}
 
 \For{$i \gets s$ \textbf{to} $f$}{\label{alg:combineparallel:for}
 Find the possible edge divisions for $i$, \label{alg:combineparallel:div}
 $\mathcal{I}=\left\{\left(j_{1},j_{2}\right): lst_{(1,j_1)}{\cdot}\eta+lst_{(2,j_2)}{\cdot}\eta=i \right\}$\;
 Find the cost maximizing division\;
    $\left(opt_1,opt_2\right)=\argmax_{\left(j_{1},j_{2}\right) \in \mathcal{I}} 
   \min\left\{lst_{(1,j_1)}{\cdot}\ell,lst_{(2,j_2)}{\cdot}\ell\right\}$\; \label{alg:combineparallel:maxdiv}
 $ \ell_{i} \gets \min\{lst_{(1,opt_1)}{\cdot}\ell,lst_{(2,opt_2)}{\cdot}\ell\}$  \;
 Create the new element $\alpha$, 
   	$\left(\alpha{\cdot}\eta, \alpha{\cdot}\ell, \alpha{\cdot}p1, \alpha{\cdot}p2\right) \gets $   
   	$ \left( i, \ell_{i}, lst_{(1,opt_1)}, lst_{(2,opt_2)}\right)$\;
   	Insert $\alpha$ in list $lst$\;
 }
 \caption{COMBINEPARALLEL}
 \label{alg:combineparallel}
\end{algorithm}

\subsection{Optimality and time complexity of Algorithm SolMINTB}

\subsubsection{Proof outline:}
The proof of Theorem \ref{thm:correct} which states that the proposed algorithm solves the $\mctp$ problem in $\sepr$ graphs in polynomial time, is broken down in Lemma \ref{lemm:SP}, Lemma \ref{lemm:optplacetoll} and Theorem \ref{thm:optmakelist}. The common theme in the proofs is the use of an inductive reasoning starting from the base case of parallel link networks, which seems natural given the parse tree decomposition. 
Lemma \ref{lemm:SP} gives a monotonicity property of the number of edges required to induce length $\ell$ in a $\sepr$ graph guiding us to induce the length of maximum used path to obtain an optimal solution. 

The key Theorem \ref{thm:optmakelist} is essentially the generalization of the ideas used in the parallel link network to $\sepr$ graphs. It proves that the lists created by Algorithm \ref{alg:mctp} follow three desired properties. 1) The maximality of the `edge-length' pairs in a list, i.e. for any `edge-length' pair $(\eta,\ell)$ in $lst_v$ it is not possible to induce a length greater than $\ell$ in $G_v$ using at most $\eta$ edges. 2) The `edge-length' pairs in a list follows an increasing length order which makes it possible to locate the optimal solution efficiently. 3) Finally the local optimality of a list at any level of the parse tree ensures that the `series' or `parallel' combination preserves the same property in the new list. 

In Lemma \ref{lemm:optplacetoll} it is proved that the appropriate tolls on the edges can be placed provided the correctness of Theorem \ref{thm:optmakelist}. The basic idea is while traversing down the parse tree at each node we induce the required length in a locally optimal manner. Finally, in the leaf nodes the tolls are placed on the edges and the process inducing a given length is complete. Exploiting the linkage between the list in a specific node and the lists in its children we can argue that this local optimal solutions lead to a global optimal solution.    

Finally, in our main theorem, Theorem \ref{thm:correct}, combining all the elements we prove that the proposed algorithm solves $\mctp$ optimally. In the second part of the proof of Theorem \ref{thm:correct}, the analysis of running time of the algorithm is carried out. The creation of the list in each node of the parse tree takes $\mathcal{O}(m^2)$ time, whereas the number of nodes is bounded by $\mathcal{O}(m)$, implying that Algorithm \ref{alg:mctp} terminates in $\mathcal{O}(m^3)$ time. Here $m$ is the number of edges in the $\sepr$ graph $G$. 

\subsubsection{Proof of correctness:}
In what follows we prove the results following the described outline.
\begin{lemma}
In an $l$-instance $\mathcal{S}$, with $\sepr$ graph $G$ and maximum used  $(s,t)$ path length $\ell_{max}$, any length $L$ can be induced in $G$ if and only if $L\geq \ell_{max}$. Moreover if length $L$ is induced optimally with $T$ edges then length $\ell_{max}\leq \ell \leq L$ can be induced optimally \emph{with $t \leq T$ edges}. 
\label{lemm:SP}
\end{lemma}
\begin{proof}

We prove the lemma using induction on the height of a parse tree decomposition of $G$. 

\vspace{-7 pt}
\paragraph{Base Case} The base case is $G$ with height $1$ which is equivalent to a parallel link network. 
Let the $m$ edges of $G$ are in ascending order according to their lengths, i.e. $\ell_1\leq \ell_2 \leq \dots \leq \ell_m$ with the used edges listed as $\{e_1,e_2,\dots, e_u\}$,  $\ell_{u+1}>\ell_u$. From the definition of inducing $L$ we know it is impossible to induce $L<\ell_u$ whereas any length $L \geq \ell_u$ can be induced optimally with the edges  in
$S= \left\{j \in \{1,\dots, m\}: L > \ell_j \right\}$.  Let $S_L$ and $S_{\ell}$ be the set of edges required to induce length $L$ and $\ell$, respectively. Our base case holds from the observation that if $\ell \leq L$ then $S_{\ell}\subseteq S_L$ and consequently $|S_{\ell}|=t\leq T=|S_L|$. 

\vspace{-7 pt}
\paragraph{Induction step} The induction hypothesis is the stated lemma holds for all $\sepr$ graphs with parse tree decomposition of height less or equal to $k$. Consider $G$ to be a $\sepr$ graph with parse tree decomposition of height $(k+1)$. The root node of the parse tree has two children, namely $G_1$ and $G_2$, with height less or equal to $k$. $G_1$ and $G_2$ can be combined either in  parallel or in series. Let $\ell_{max}$, $\ell_{max,1}$ and $\ell_{max,2}$ be the lengths of maximum used $(s,t)$  paths in $G$, $G_1$ and $G_2$, respectively. 

\vspace{-7 pt}
\paragraph{Parallel Combination} We induce $L$ between the common start node $s$ and end node $t$ of the parallel subgraphs, $G_1$ and $G_2$, so as to induce length $L$ in $G$. Clearly, from induction hypothesis it is possible to induce any length $L \geq max\{\ell_{max,1}, \ell_{max,2}\}=\ell_{max}$ whereas we can not induce length less than $\ell_{max}$.
Further, let $L$ and $\ell$ be optimally induced in $G_i$ with $T_i$ and $t_i$ edges respectively, for $i\in \{1,2\}$. Hence we require $T=T_1+T_2$ edges to optimally induce $L$ and $t=t_1+t_2$ edges to optimally induce $\ell$. According to induction hypothesis $\ell\leq L$ implies $T_i \geq t_i$ for both $i=1$ and $2$, which proves that $t\leq T$. 

\vspace{-7 pt}   
\paragraph{Series Combination}
In order to induce length $L$ in $G$ we split $L$ into $L_1$ \& $L_2$ and induce them in $G_1$ and $G_2$. Lengths $L_i$ can be induced in $G_i$ if and only if $L_i\geq \ell_{max,i}$ for $i\in\{1,2\}$. Clearly any length $L\geq \ell_{max,1}+\ell_{max,2}=\ell_{max}$ can be induced in $G$, $\ell_{max}$ being the minimum possible. 
Further, let $L\geq \ell$ with $L$ induced optimally with $T$ edges in $G$, and $\ell$ with $t$ edges. The optimal assignment for length $L$ is, $T_1$ edges used to induce $L_1$ on $G_1$ and $T_2$ edges to induce $L_2$ on $G_2$ with $T=T_1+T_2$ and $L=L_1+L_2$. Supposing $T<t$, i.e. we induce $L$ with less number of edges, we will reach a contradiction. In the above case consider the following assignment to induce $\ell$ in $G$: We use $t_1'$ edges to induce $min\{L_1,\ell-\ell_{max,2}\}$ on $G_1$ optimally, whereas using $t_2'$ edges we optimally induce $max\{\ell-L_1,\ell_{max,2}\}$ on $G_2$. If $L_1\leq \ell-\ell_{max,2}$ then $t_1'= T_1$ and due to the induction hypothesis we have $T_2 \geq t_2'$ as $L_2 \geq \ell-L_1$. Therefore, in the new assignment we require $t'=t_2'+T_1$ edges with $t'\leq T < t$. This leads to a contradiction to $t$ being the number of edges required in an optimal assignment. On the contrary when $L_1 > \ell-\ell_{max,2}$ we have $t_1' \leq T_1$ in $G_1$. Also we have $t_2' \leq T_2 $ in $G_2$ as the induced length $\ell_{max,2}$ is the minimum possible length. This leads to the same contradiction as before, 
making the claim in the lemma valid for the level $(k+1)$ graph $G$. By the principle of induction the lemma holds for any $\sepr$ graph $G$. \qed
 
\end{proof}

\begin{minipage}{0.7\linewidth}
Note: The above lemma breaks in general graphs. As an example, in the graph in Figure \ref{fig:counter} to induce a length of $3$ we require $3$ edges, whereas to induce a length of $4$ only $2$ edges are sufficient.
\end{minipage} 
\hfill
\begin{minipage}{0.3\linewidth}\centering
\includegraphics[width=0.45\linewidth]{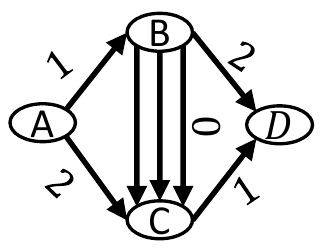}
\captionof{figure}{Counter example}
\label{fig:counter}
\end{minipage}

\begin{theorem}
Let $\mathcal{S}$ be an $l$-instance and $G$ be the associated $\sepr$ graph with parse tree representation $T$. For every node $v$ in $T$, let the corresponding $\sepr$ network be $G_v$ and $\ell_{max,v}$ be the length of the maximum used path from $s_v$ to $t_v$. Algorithm \ref{alg:makelist} creates the list, $lst_v$, with the following properties.
\begin{enumerate}
\item For each `edge-length' pair $(\eta_i,\ell_i)$, $i=1$ to $m_v$, in $lst_v$, $\ell_i$ is the maximum length that can be induced in the network $G_v$ using at most $\eta_i$ edges.
\item For each `edge-length' pair $(\eta_i,\ell_i)$ in the list $lst_v$, we have the total ordering, i.e. $\eta_{i+1}=\eta_{i}+1$ for all $i=1$ to $m_v-1$, and $\ell_{max,v}=\ell_1 \leq \ell_2 \leq \dots \leq \ell_{m_v}=\infty$. 
\item In $G_v$ , length $\ell$ is induced by minimum $\eta_{\hat{i}}$ edges if and only if $\ell\geq \ell_1$ and $\hat{i}= \argmin\{\eta_j:(\eta_j,\ell_j)\in lst_v \wedge \ell_j\geq \ell\}$.
\end{enumerate}
\label{thm:optmakelist}
\end{theorem}
\begin{proof}

We first prove property $3$ given property $1$ and $2$. Suppose the minimum number of edges needed to induce length $\ell\geq \ell_1$ is $\eta_{opt}$. From Lemma \ref{lemm:SP} such an $\eta_{opt}$ exists as 
$\ell_1=\ell_{max,v}$ from property $2$. Moreover Lemma \ref{lemm:SP} says $\eta_{opt}\leq \eta_j$ for all `edge-length' pairs $(\eta_j,\ell_j)\in lst_v$ with $\ell_j\geq \ell$. Therefore, $\eta_{opt} \leq \eta_{\hat{i}}$, where $\hat{i}$ is as defined in the theorem. 
Property $1$ implies that for all `edge-length' pairs $(\eta_j,\ell_j)\in lst_v$ such that  $\ell> \ell_j$, the required number of edges $\eta_{opt}>\eta_j$. Note the list $lst_v$ has total ordering due to property $2$ and in list $lst_v$ the first element to have cost greater or equal to $\ell$ is $\hat{i}$. Therefore, the minimum number of edges required to induce $\ell$ is $\eta_{opt}= \eta_{\hat{i}}$. We prove property $1$ and $2$ using induction on the levels of the parse tree of $G$.


\vspace{- 7 pt}
\paragraph{Base Case} Consider any leaf node $v$ in $T$, node $v$ corresponds to a parallel link subgraph of $G$, $G_v$. The base case easily follows from the discussion regarding the algorithm for parallel link networks.

\vspace{- 7 pt}
\paragraph{Induction step} Let the properties $1$, $2$ and consequently property $3$ hold for every node upto level $k$ from the leaf. Consider the graph $G_v$, the graph corresponding to a level $(k+1)$ node $v$. The subgraphs $G_1$ and $G_2$, both at level at most $k$, have the lists $L_1$ and $L_2$, respectively. To induce length $L$ in $G_v$ we need to induce $L$ on both $G_1$ and $G_2$, if $G_v$ is the parallel combination of the two subgraphs. However, in the series  case we induce $\ell'$ on $G_1$ and $L-\ell'$ on $G_2$.

We proceed by constructing the feasible range $\{s,\dots,f\}$ of the list $lst_v$ such that for any `edge-length' pair $(\eta,\ell)$, $s\leq \eta\leq f$. 
The series and the parallel combinations manifest different ranges.
In lines \ref{alg:combineseries:s}, \ref{alg:combineseries:f} in Algorithm \ref{alg:combineseries} we calculate the feasible range for series combination. The starting `edge-length' pair will induce $\ell_{max,1}$ in $G_1$ and $\ell_{max,2}$ in $G_2$. Whereas, to get the last pair in $lst_v$ we induce $\infty$ in one of the subgraph and the minimum length in the other. The combination that uses minimum number of edges of the two possiblities, is used.   
For the parallel combination in line \ref{alg:combineparallel:s} of Algorithm \ref{alg:combineparallel} we calculate the minimum number of edges in the list $lst_v$. From Lemma \ref{lemm:SP} we cannot induce any length less than $\ell_{max,v}=\max\{\ell_{max,1},\ell_{max,2}\}$. Also due to property $3$ of $G_1$ and $G_2$,  $s_1$ and $s_2$, as given in line \ref{alg:combineparallel:s:sub}, are the minimum number of edges required to induce $\ell_{max,v}$ in $G_1$ and $G_2$, respectively. To induce the maximum length of $\infty$ in $G$ we require the maximum number of edges in both the lists as in line \ref{alg:combineparallel:f}. 
This proves we have $\ell_1=\ell_{max,v}$ and $\ell_{m_v}=\infty$.

\vspace{- 7 pt}   
\paragraph{Property $1$} Let $(\eta,\ell)$ be an `edge-length' pair in $lst_v$ and consider any feasible division, $\eta'$ in $G_1$ and $\eta-\eta'$ in $G_2$. A division is feasible if $s_1\leq \eta'\leq f_1 $ and $s_2\leq \eta- \eta'\leq f_2$. For all such divisions we have entries $(\eta',\ell_1)$ in $lst_1$ and $(\eta-\eta',\ell_2)$ in $lst_2$ due to property $2$ of $G_1$ and $G_2$. For this particular division the maximum length that can be induced in $G_1$ is $\ell_1$ and in $G_2$ is $\ell_2$. This holds from the property $1$ of both $lst_1$ and $lst_2$.  

Next consider the series and parallel cases separately. The maximum length that can be induced is $\min\{\ell_1,\ell_2\}$ in the parallel combination as it is required to induce length $\ell$ in both the subgraphs. Whereas, in the series combination the maximum length induced is $(\ell_1+\ell_2)$ for the given division. Finally, in line \ref{alg:combineparallel:maxdiv} of Algorithm \ref{alg:combineparallel} and  line \ref{alg:combineseries:maxdiv} of Algorithm \ref{alg:combineseries}, we find the length maximizing division among all feasible divisions. 
Consequently $\ell$ is the maximum length induced by at most $\eta$ edges. Property $1$ holds.  

\vspace{- 7 pt}
\paragraph{Property $2$}  For every $\eta$ in the range of $lst_v$ the feasible set $\mathcal{I}$, as computed in line \ref{alg:combineparallel:div} in Algorithm \ref{alg:combineparallel} and line \ref{alg:combineseries:div} in Algorithm \ref{alg:combineseries}, is non-empty implying $\eta_{j+1}=\eta_{j}+1$ in $lst_v$ for all $j$. This follows as due to property $2$ of $G_1$ and $G_2$, we always have at least one feasible division for any $\eta\in\{s,\dots,f\}$.

Now consider two pairs $(\eta_i,\ell_i)$ and $(\eta_j,\ell_j)$ in $lst_v$ with $\eta_j>\eta_i$. We give a feasible division for $\eta_j$ and show that the maximum length that we can induce with this division is not less than $\ell_i$. Specifically, first using $\eta_i$ edges induce $\ell_i$ in $G_v$ while using the remaining $(\eta_j-\eta_i)$ edges arbitrarily in $G_1$ and $G_2$. Let this division induce $\ell$ in $G$. Due to the fact that list $lst_1$ and $lst_2$ has total ordering, adding edges can never decrease the maximum length induced. Therefore, we have $\ell \geq \ell_i$. Further $\ell_j$ is the maximum length that can be induced using $\eta_j$ edges for any division. Clearly, $\ell_j\geq \ell\geq \ell_i$ and property $2$ holds. 

To conclude the proof, the properties $1$ and $2$ holds for $lst_v$ and the inductive hypothesis extends to all the nodes in $T$, due to the induction principle. \qed
\end{proof}
\begin{lemma}
In an $l$-instance $\mathcal{S}$ with $\sepr$ graph $G$, suppose we are given lists $lst_{v}$, for all nodes $v$ in the parse tree $T$ of $G$, all of which satisfy properties $1$, $2$ and $3$ in Theorem \ref{thm:optmakelist}. Algorithm \ref{alg:placetoll} induces any length $\ell_{in}\geq \ell_1$ optimally in $G$, where $\ell_1$ is the length of the first `edge-length' pair in $lst_r$, $r$ being the root node of $T$. Moreover, it specifies the appropriate tolls necessary for every edge in the network.
\label{lemm:optplacetoll}
\end{lemma}
\begin{proof}
The list $lst_r$ satisfy property $3$ for $\ell_{in} \geq \ell_1$. Therefore, the minimum number of edges required to induce $\ell_{in}$ is $\eta_{opt}$ where $opt=\argmin\{\eta_j:(\eta_j,\ell_j)\in lst_r \wedge \ell_j\geq \ell_{in}\}$ and the `edge-length' pair is $(\eta_{opt},l_{opt})$. Using pointers $lst_r{\cdot}p1$ and $lst_r{\cdot}p2$, we maintain a proper linkage  to the appropriate `edge-length' pairs in the list of its children. Let the `edge-length' pairs be $(\eta_{opt_1},\ell_{opt_1})$ and $(\eta_{opt_2},\ell_{opt_2})$ in $G_1$ and $G_2$ respectively. We have $\eta_{opt}=\eta_{opt_1}+\eta_{opt_2}$ and either $l_{opt}=l_{opt_1}=l_{opt_2}$ in parallel or   $l_{opt}=l_{opt_1}+l_{opt_2}$  in series.

Algorithm \ref{alg:placetoll} induces $\ell_{in}$ optimally, provided the length $\ell_{in}$ is split in $\ell_1$ and $\ell_2$ properly. Call a split $(\ell_1,\ell_2)$ proper if it follows the following two properties 1) $\ell_{in}=\ell_1+\ell_2$ for series  and $\ell_{in}=\ell_1=\ell_2$ for parallel nodes; 2) $\ell_i \leq \ell_{opt_i} $ for $i\in\{1,2\}$. The two conditions are both necessary and together sufficient for the optimal assignment of tolls, which follow from Lemma \ref{lemm:SP} and property $3$. The necessity of the first condition follows from the monotonicity property in the current node and the second condition is necessary due to the monotonicity property in the two children. Whereas, the property $3$ of $lst_r$ implies the sufficiency when both hold. In the lines \ref{alg:placetoll:divlen} to \ref{alg:placetoll:divlen2} the division created is a proper division. In the parallel case it is true in the obvious way, whereas in the series case note that $\ell_1=\ell_{opt_1}$ and $\ell_2=\ell_{in}-\ell_{opt_1}\leq \ell_{opt}-\ell_{opt_1}=\ell_{opt_2}$, also $\ell_{in}=\ell_1+\ell_2$. So the division is indeed a proper division.     

Moreover at the leaf node level the algorithm sets the tolls in line \ref{alg:placetoll:set_toll} in the parallel link networks. From the discussion  regarding the parallel link network the optimality of the solution follows. \qed

\end{proof}

\begin{theorem}
Algorithm \ref{alg:mctp} solves the $\mctp$  problem optimally in time $\mathcal{O}\left(m^3\right)$ for the instance $\mathcal{G}=\{G(V,E),(\ell_e)_{e\in E},r\}$, where $G(V,E)$ is a $\sepr$ graph with $|V|=n$ and $|E|=m$.
\label{thm:correct}
\end{theorem}
\begin{proof}

In Algorithm \ref{alg:mctp} given instance $\mathcal{G}$  and an optimal flow $o$, we create  corresponding $l$-instance $\mathcal{S}=\{G(V,E),$ $\{l_e\}_{e\in E}, E_u\}$. We construct the parse tree $T$ of $G$ and invoke Algorithm \ref{alg:makelist} to compute the lists for every node in the parse tree $T$ which satisfy the properties in Theorem \ref{thm:optmakelist}. 
Algorithm \ref{alg:placetoll} next induces length $\ell=\ell_{max}$ optimally with $|C|$ number of edges, where $C=\{e: \theta_e >0\}$ (see Lemma \ref{lemm:optplacetoll}). Let $OPT$ be the required number of edges in an optimal solution to the $\mctp$ problem. From Lemma \ref{lemm:induceequiv} $OPT$ induces length $L\geq \ell$. But due to the Lemma \ref{lemm:SP},  $\ell\leq L$ implies $|C|\leq OPT$. Therefore, $|C|=OPT$ and the calculated toll $C$ is an optimal solution for $\mctp$.

In the remaining part we derive the run-time of our algorithm and show it is polynomial in network size. The running time for creating a parse tree $T$ of size $\mathcal{O}(m)$ for $G$ is $\mathcal{O}(m)$, while the calculation of the  $l$-instance $\mathcal{S}$  and $\ell_{max}$ also takes $\mathcal{O}(m)$. 

 The number of runs of Algorithm \ref{alg:makelist}, in the recursive routine, when called from the root of $T$ is exactly equal to the number of nodes in $T$. Therefore, in order to bound the run time for Algorithm \ref{alg:makelist} first we need an upper bound on the size of each list created. It is trivial to note that the size of list $lst_v$ is at most the number of edges in the subgraph $G_v$ rooted at node $v$. Further we have $|lst_v|\leq m$, for all $v$, as $G_v$ is a subgraph of $G$. 
Next requirement are the run-times of Algorithm \ref{alg:makelistpa}, Algorithm \ref{alg:combineseries} and  Algorithm \ref{alg:combineparallel}. 
 
 In the `leaf' nodes the time complexity for Algorithm \ref{alg:makelistpa} is dominated by the sorting step in line \ref{alg:makelistpa:order} resulting in $\mathcal{O}(m \log m)$ time requirement. Its easy to observe that Algorithm \ref{alg:combineseries} and  Algorithm \ref{alg:combineparallel} have similar structure resulting in the same time complexity. We show the time complexity of Algorithm \ref{alg:combineseries} is $\mathcal{O}\left(m^3\right)$. The time complexity for Algorithm \ref{alg:combineparallel} follows similarly. Combining these we obtain that  Algorithm \ref{alg:makelist} in line \ref{alg:mctp:makelist} of Algorithm \ref{alg:mctp} is executed in $\mathcal{O}\left(m^3\right)$ time.

  In Algorithm \ref{alg:combineseries} the number of iteration in line \ref{alg:combineseries:for} is bounded by the list size at the specific node. Moreover, for each iteration, in time linear in the list size the algorithm computes the feasible set $\mathcal{I}$ in line \ref{alg:combineseries:div}, where the total ordering of each list plays a crucial role. The execution of the maximization step in line \ref{alg:combineseries:maxdiv} occurs in linear time as well. So the overall run-time of Algorithm \ref{alg:combineseries} is $\mathcal{O}\left(m^2\right)$ and completes the above argument.
 
Algorithm \ref{alg:placetoll} when called from the root node of $T$ runs exactly once for each node in $T$. The runtime at each individual node is bounded by $\mathcal{O}(m)$ resulting in a total runtime of $\mathcal{O}\left(m^2\right)$.
In conclusion, Algorithm \ref{alg:mctp} generates optimal solution to $\mctp$ problem for $\sepr$ graphs in $\mathcal{O}\left(m^3\right)$ time.    \qed

\end{proof}


\section{Conclusion}
In this paper we consider the problem of inducing the optimal flow as network equilibrium and show that the problem of finding the minimum cardinality toll, i.e. the $\mctp$  problem,  is NP-hard to approximate within a factor of $1.1377$. Furthermore we define the minimum cardinality toll with only used edges left in the network and show in this restricted setting the problem remains NP-hard even for single commodity instances with linear latencies. We leave the hardness of approximation results of the problem open. Finally, we propose a polynomial time algorithm that solves $\mctp$ in series-parallel graphs, which exploits the parse tree decomposition of the graphs. The approach in the algorithm fails to generalize to a broader class of graphs. Specifically, the monotonicity property proved in Lemma \ref{lemm:SP} holds in series-parallel graphs but breaks in general graphs revealing an important structural difficulty inherent to $\mctp$ in general graphs. Future work involves finding approximation algorithms for $\mctp$. The improvement of the inapproximability results presented in this paper provides another arena to this problem, e.g. finding stronger hardness of approximation results for $\mctp$ in multi-commodity networks.

\section*{Acknowledgement}
We would like to thank Steve Boyles and Sanjay Shakkottai for helpful discussions. This work was supported in part by NSF grant numbers CCF-1216103, CCF-1350823 and CCF-1331863.

\bibliography{mybib}
\bibliographystyle{unsrt}

\end{document}